\documentclass[letterpaper]{article} 
\usepackage[submission]{aaai23}  
\usepackage{times}  
\usepackage{helvet}  
\usepackage{courier}  
\usepackage[hyphens]{url}  
\usepackage{graphicx} 
\usepackage{natbib}  
\usepackage{caption} 
\usepackage{algorithm}
\usepackage{algorithmic}

\usepackage{amsmath}
\usepackage{amsthm}
\usepackage{subfigure}
\usepackage{color}

\newtheorem{theorem}{Theorem}
\newtheorem{corollary}{Corollary}

\usepackage[title]{appendix}
\usepackage{newfloat}
\usepackage{listings}
\DeclareCaptionStyle{ruled}{labelfont=normalfont,labelsep=colon,strut=off} 
\floatstyle{ruled}
\newfloat{listing}{tb}{lst}{}
\floatname{listing}{Listing}
\pdfoutput=1 

\title{An Efficient Approach to the Online Multi-Agent Path Finding Problem \\ by Using Sustainable Information }
\author{
    \textsuperscript{\rm 1, \rm 3}Mingkai TANG,
    \textsuperscript{\rm 1, \rm 6}Boyi LIU,
    \textsuperscript{\rm 1, \rm 4}Yuanhang LI,
    \textsuperscript{\rm 1, \rm 2, \rm 5}Hongji LIU,
    \textsuperscript{\rm 1, \rm 2, \rm 7}Ming LIU,
    \textsuperscript{\rm 1, \rm 8}Lujia WANG\\
}
\affiliations{
    \textsuperscript{\rm 1}The Hong Kong University of Science and Technology\\
    \textsuperscript{\rm 2}The Hong Kong University of Science and Technology (Guangzhou)\\

    \{\textsuperscript{\rm 3}mtangag, \textsuperscript{\rm 4}yliog, \textsuperscript{\rm 5}hliucq\}@connect.ust.hk, \textsuperscript{\rm 6}by.liu@ieee.org, \{\textsuperscript{\rm 7}eelium, \textsuperscript{\rm 8}eewanglj\}@ust.hk
}

\usepackage{bibentry}

\begin{document}

\maketitle

\begin{abstract}
Multi-agent path finding (MAPF) is the problem of moving agents to the goal vertex without collision. In the online MAPF problem, new agents may be added to the environment at any time, and the current agents have no information about future agents. The inability of existing online methods to reuse previous planning contexts results in redundant computation and reduces algorithm efficiency. Hence, we propose a three-level approach to solve online MAPF utilizing sustainable information, which can decrease its redundant calculations. The high-level solver, the Sustainable Replan algorithm (SR), manages the planning context and simulates the environment. The middle-level solver, the Sustainable Conflict-Based Search algorithm (SCBS), builds a conflict tree and maintains the planning context. The low-level solver, the Sustainable Reverse Safe Interval Path Planning algorithm (SRSIPP), is an efficient single-agent solver that uses previous planning context to reduce duplicate calculations. Experiments show that our proposed method has significant improvement in terms of computational efficiency. In one of the test scenarios, our algorithm can be 1.48 times faster than SOTA on average under different agent number settings.
\end{abstract}

\section{Introduction}
The multi-agent path finding problem (MAPF) is finding paths for a set of agents to move from their starting vertex to the goal vertex without collision. MAPF has a wide practical application, such as aircraft towing
vehicles \cite{morris2016planning}, warehouse robots \cite{wurman2008coordinating}, video games \cite{ma2017feasibility} and urban road networks \cite{choudhury2022coordinated}.


For MAPF, most works assume that the environment can be fully captured before the system runs \cite{salzman2020research, stern2019multi, ma2017overview}. Under this assumption, the solution can be calculated in advance, and the agent only needs to take actions along the pre-calculated plan at runtime. These path finding problems are referred to as offline MAPF. However, in practice, the assumption is not always guaranteed. During the running of a system, new agents might appear in the system suddenly, and agents need to do replanning to fit the new situation. Recently, the online MAPF problem  \cite{vsvancara2019online} was proposed. It is assumed that new agents may be added to the environment at any time, and the current agents have no information about future agents. All agents in the environment need to do replanning to fit the new situation.

Online MAPF is a problem of great practical importance. For example, in a real-world warehouse system, the robot frequently enters and exits the work area due to factors such as charging or malfunction. Moreover, the time point at which the robot re-enters the work area is unpredictable. The agents in the warehouse need to adjust their path due to the appearance of a new agent. 

\begin{figure}[t]
\centering
\includegraphics[width=0.30\textwidth]{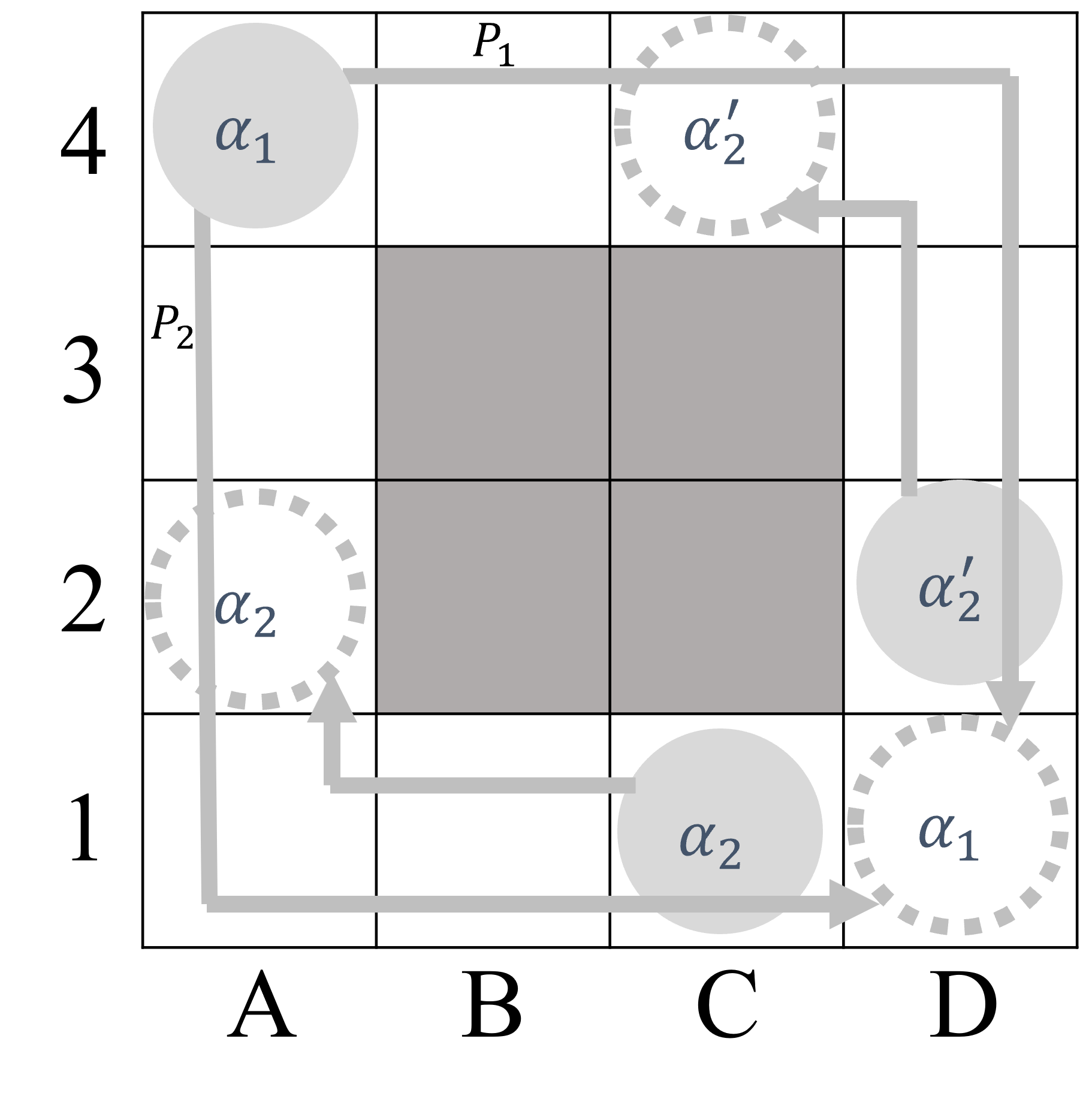} 
\caption{An example of two online MAPF instances. In both instances, $a_1$ appears $A4$ at time point $0$ and needs to travel to $D1$. $P_1$ and $P_2$ are two paths with equal costs for $a_1$. In the first instance, $a_2$ appears in $C1$ at time point $1$, and its goal is $A2$. In the second instance, $a_2'$ appears in $D2$ at time point $1$, and its goal is $C4$.}
\label{fig:online_mapf_instance}
\end{figure}

Optimally solving the offline MAPF problem is NP-hard \cite{yu2013structure, ma2016multi}. Compared with the offline MAPF, which can calculate the paths before the system runs, the online MAPF further needs to calculate high-quality paths for all agents in real-time when new agents appear. \citeauthor{vsvancara2019online} proposed several methods to solve the online MAPF problem.
The Replan Single algorithm searches an optimal path for each new agent when they appear, while all paths of the old agents remain the same. The Replan Single Group algorithm jointly plans for all new agents appearing at the same time, and the new agents' path plan cannot affect the old agents' path plan. These two algorithms can execute fast, but the solutions are not optimal. The Replan All algorithm replan for all agents without considering existing path plans when new agents appear, and it can get high-quality solutions. However, each iteration of planning in the Replan All algorithm has a high computational complexity when the number of agents becomes large. 

Considering reusing the information in previous planning iterations to reduce the running time, we propose an efficient algorithm. In this paper, we refer to this kind of information as sustainable information or planning context. 
Our method consists of three levels of algorithms.
We name the high-level algorithm as the Sustainable Replan algorithm(SR). It simulates the environment and maintains the whole planning context. The middle-level algorithm, the Sustainable Conflict-Based Search Algorithm (SCBS), is called by SR for searching the multi-agent path planning solution based on current information. SCBS uses the Sustainable Reverse Safe Interval Path Planning algorithm (SRSIPP), which is the low-level solver, for single-agent planning. Given that each planning iteration for a single agent has the same goal point, but different starting points, SRSIPP searches the path in the backward direction (from goal point to starting point) to reuse the previous planning information.

The main contributions of this paper are as follows.
\begin{enumerate}
    \item We propose a three-level approach for reusing previous planning context to reduce the running time for the online MAPF.
    \item We prove the completeness and the snapshot optimality of our approaches.
    \item We performed detailed algorithm performance comparison experiments with SOTA. The average acceleration rate relative to the SOTA can reach up to 1.48.
\end{enumerate}

\section{Problem Definition}
The definition of the online multi-agent path finding problem is that given a directed graph $G(V, E)$, and a set of $k$ agents $a_1$, $a_2$, $a_3$ ... $a_k$, find a collision-free path for each agent. The agent $a_i$ is described by the triplet $(t^s_i, v^s_i, v^g_i)$, meaning agent $a_i$ appears in the starting vertex $v^s_i \in V$ in time point $t^s_i$ and its goal is the vertex $v^g_i \in V$. In this paper, we call agent $i$ \textit{starts} at $t^s_i$. Without loss of generality, we assume $0 \leq t^s_1 \leq t^s_2 \leq ... \leq t^s_k$. Specially, agents whose start time point is $0$ can be seen as agents already in the scene before the environment starts to run, and we refer to the planning for these agents as the \textit{offline} part of the online MAPF problem. In contrast, we refer to the planning after the system starts as the \textit{online} part. In the beginning, all agents plan their path from their own start vertex to the goal vertex, while they do not know any information about the agents that will start in the future. After that, agents follow their own plan at each time. When it comes to the time point when new agents start, all agents need to replan their paths considering the new input of the online MAPF problem. Let $m$ be the number of time points when new agents start, and $t^{new}_1$, $t^{new}_2$ ... $t^{new}_m$ be the corresponding time point sequence. $m$ might be smaller than $k$ because 
there may be more than one agent starting at the same time point. The solution to an online MAPF problem is defined as a sequence of valid plans $\Pi = \left \langle \pi^0, \pi^1, \pi^2 ... \pi^m \right \rangle$, where
  $\pi^j$ is a collection of all path plans at $t^{new}_j$.  Let  $p^j_{i}$ be the path plan of agent $i$ in $\pi^j$, and $p^j_i[t]$ be the vertex of agent $i$ in time point $t$ in $\pi^j$. We define $p^j_i[t_l:t_r]$ as the concatenation of the path plan of agent $i$ from time point $t_l$ to time point $t_r$, i.e. $p^j_i[u:v] = p^j_i[u] \circ p^j_i[u+1] \circ ... \circ p^j_i[v - 1]$, where $\circ$ is the concatenation operator. The execute plan of agent $i$ is defined as $Ex_i[\Pi] = p^1_i[t^{new}_1:t^{new}_2]   \circ  p^2_i[t^{new}_2:t^{new}_3] \circ ... \circ p^{m-1}_j[t^{new}_{n-1}:t^{new}_n] \circ p^{n}_j[t^{new}_n:\infty]$, showing the actual path of agent $i$.

In this paper, we focus on the variant: 
\begin{itemize}
\item  We assume that the agent starts in the garage, which means that the new agent can choose to enter the start vertex at the start time or later. Before they enter, they wait in the garage and do not conflict with other agents. In addition, we use the setting of disappearing at the goal vertex. Under these two assumptions, the problem is always solvable if the offline part is solvable and each agent has a path from its initial location to its goal location, as proved in Proposition 2 in \cite{vsvancara2019online}. Although we use these two assumptions, our proposed method can be easily extended to other assumptions at the start and goal. 
\item  We only consider \textit{vertex conflict} and \textit{edge conflict}. Two agents collide iff they occupy the same vertex or cross the same edge in opposite directions at the same time. 
\end{itemize}
\begin{figure*}[t]
\centering
\includegraphics[width=0.75\textwidth]{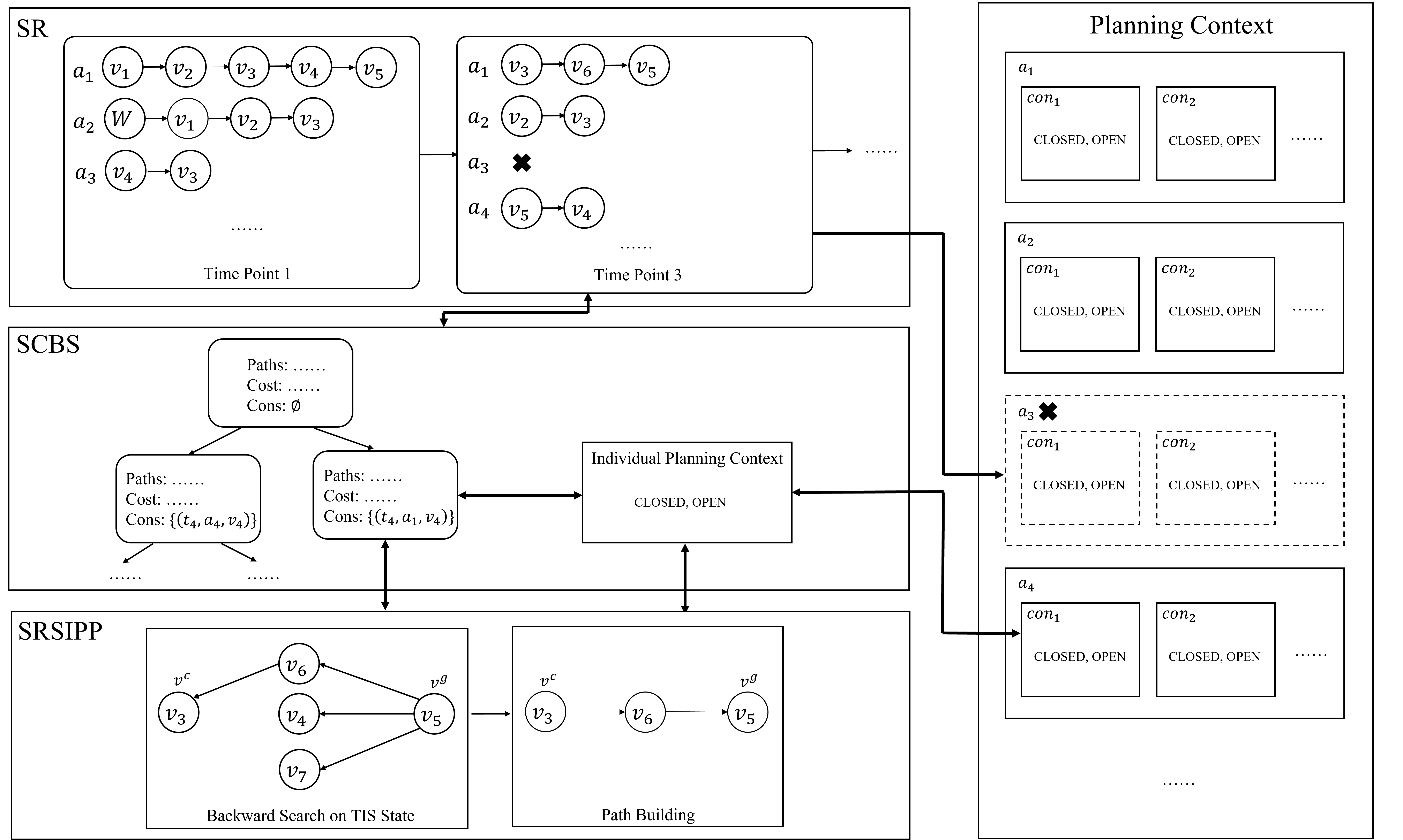} 
\caption{The architecture of the three-level approach. The SR algorithm is the high-level solver, simulating the environment and managing the planning context. The SCBS algorithm is the middle-level solver, which builds a conflict tree and extracts the individual planning context. The low-level solver, the SRSIPP algorithm, uses backward search on the TIS state for single-agent path planning.}
\label{fig:system}
\end{figure*}
Two objectives are commonly used for offline MAPF problems, minimizing \textit{makespan} and minimizing \textit{sum-of-cost (SOC)}. Makespan is the maximum complete time above all agents. However, minimizing makespan is improper for the online MAPF problem because new agents will continuously be added to the environment, and the later added agent will more likely affect the objective. SOC is the summation of the cost of all agents' path plans. However, two online MAPF solvers are not comparable in SOC directly because SOC cannot measure the exact quality of the two solvers. For example, two solvers $s_1$ and $s_2$ are used to solve the instances in Figure \ref{fig:online_mapf_instance}. At time point $0$, $a_1$ starts in $A4$. It has two paths with the same cost to its goal $D1$. Assume $s_1$ choose $P_1$ and $s_2$ choose $P_2$. Now considering for $s_1$, $a_1$ goes to $B4$ at time point $1$. In the first instance, $a_2$ appears, the path of $a_1$ will not be affected, and it will continuously follow the path [$B4$, $C4$, $D4$, $D3$, $D2$, $D1$] with length 6. However, in another instance, $a_2'$ appears, the path of $a_1$ will make a detour [$B4$, $A4$, $A3$, $A2$, $A1$, $B1$, $C1$, $D1$] with length $8$. The symmetrical situation will appear on $s_2$. We cannot say $s_1$ is better than $s_2$ or not because the actual cost depends on the future agents, which is unpredictable at early time points. Using SOC directly can not judge the quality of the solver. We define a solver as a \textit{snapshot optimal solver} if the solver can get optimal paths in terms of SOC, assuming no new agent will appear in the future. A  snapshot optimal solver is better than a non-optimal solver in solution quality.

\section{Methodology}
Our approach is a three-level method. Figure \ref{fig:system} shows the architecture of the method. The high-level solver is the Sustainable Replan algorithm (SR), which simulates the environment and maintains the planning context of all agents. The Sustainable Conflict-Based Search algorithm (SCBS), the middle-level solver, plans the optimal path for multi-agents and manages the planning context. The low-level solver, Sustainable Reverse Safe Interval Path Planning (SRSIPP), solves a single-agent problem under a set of constraints.

\subsection{Sustainable Replan Algorithm}
SR algorithm is the highest solver. It can simulate the scene and maintain the planning context sustainably. When one or more agents appear, the algorithm will do replanning for all agents. Figure \ref{fig:system} shows an example for the SR algorithm. The 'W' in the figure means the action is waiting in the garage.

We define $pc$ as the planning context, a two-level hash table, to save all the planning context. Its keys are the agent's id and all constraints on this agent, while its values can be determined by its lower-level solvers. We will describe the specific planning context in the later subsections.

The pseudo-code is shown in Algorithm \ref{alg:sr}. Let $t^c$ be the current time point, $v^c$ be the current vertex of the agent located, $A$ be the agent set that has started, $A^+$ be the new agent set appearing in time point $t^c$, and $Ex$ be the execution plan. In lines 1-8, the environment is simulated to time point $t^c$. If an agent $a$ reaches its goal before time point $t^c$, all related elements in $A$ and $pc$ will be removed. Otherwise, the current vertex $v^c$ of the agent $a$ will be obtained from the previous execute plan $Ex$. In line $10$, the SCBS algorithm calculates the optimal path plan $p$. In line $11$, the execution plan is updated by the $p$, deleting the path from time point $t^c$ and concatenating the new path plan to the execution plan. 
\begin{algorithm}[tb]
\caption{Sustainable Replanning}
\label{alg:sr}
\textbf{Input}: original agent set $A$, new agent set $A^+$, execute plan $Ex$, current time point $t^c$, planning context $pc$ 

\begin{algorithmic}[1] 
\FOR{agent $a$ in $A$}
\IF{$a$ reach goal before time $t^c$}
\STATE $A$ $\leftarrow$ $A$ $\backslash$ \{$a$\}
\STATE Remove all planning context of $a$ in $pc$.
\ELSE
\STATE Update $a.v^c$ by $Ex$.
\ENDIF
\ENDFOR
\STATE  $A \leftarrow A \cup A^+$.
\STATE  $p, pc \leftarrow SCBS(A, t^c, pc)$    // Algorithm 2
\STATE  Update $Ex$ by $p$.
\RETURN $A$, $Ex$, $pc$
\end{algorithmic}
\end{algorithm}

\begin{figure*}[t]
\centering 
    \subfigure[Forward search on TS states.]{
       \centering
       \label{fig:forward_search_ts}
        \includegraphics[width=0.33\textwidth]{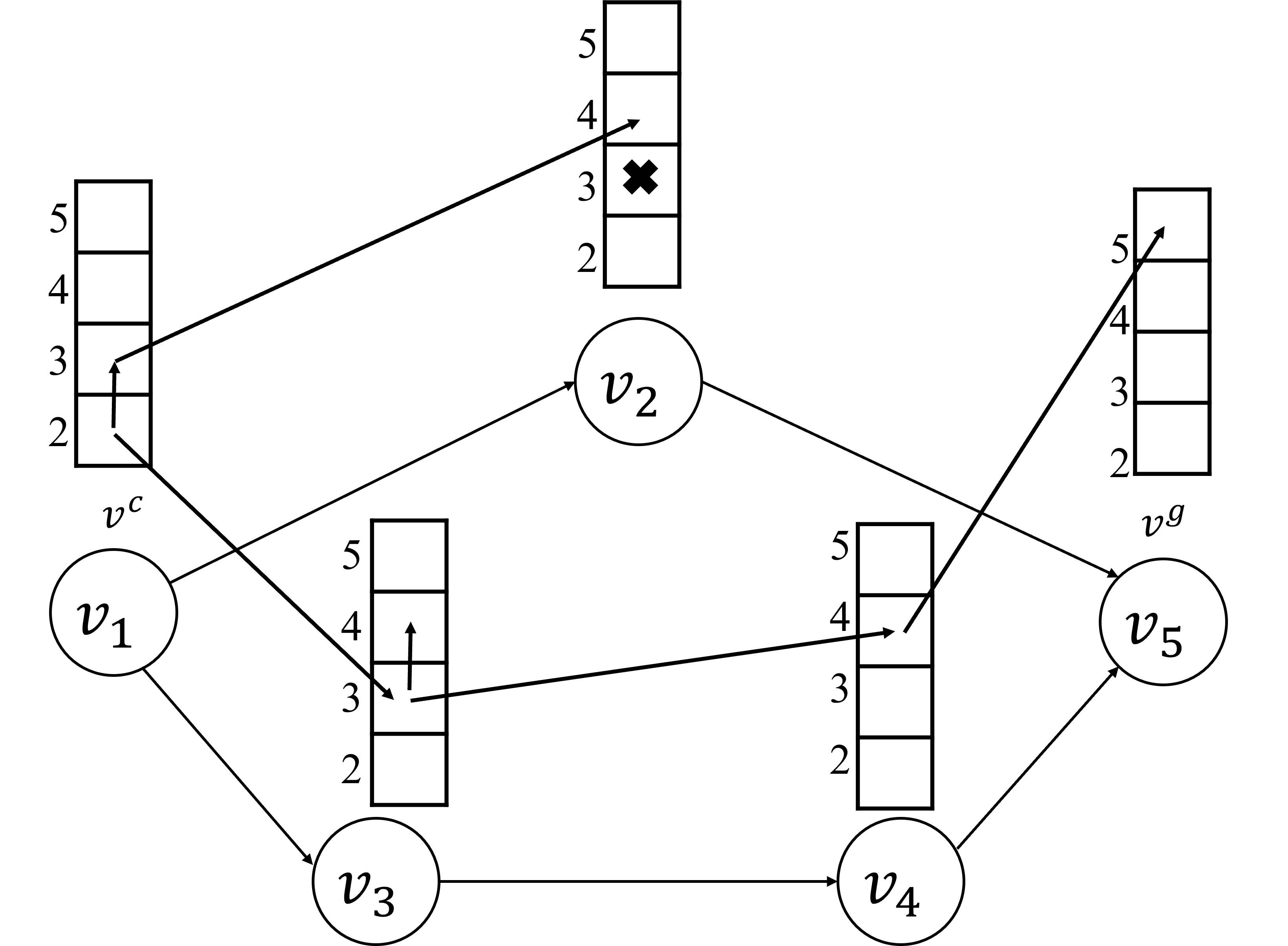}
    }
    \hspace{0.09\textwidth}
    \subfigure[Backward search on TS states.]{
        \centering
        \label{fig:backward_search_ts}
        \includegraphics[width=0.33\textwidth]{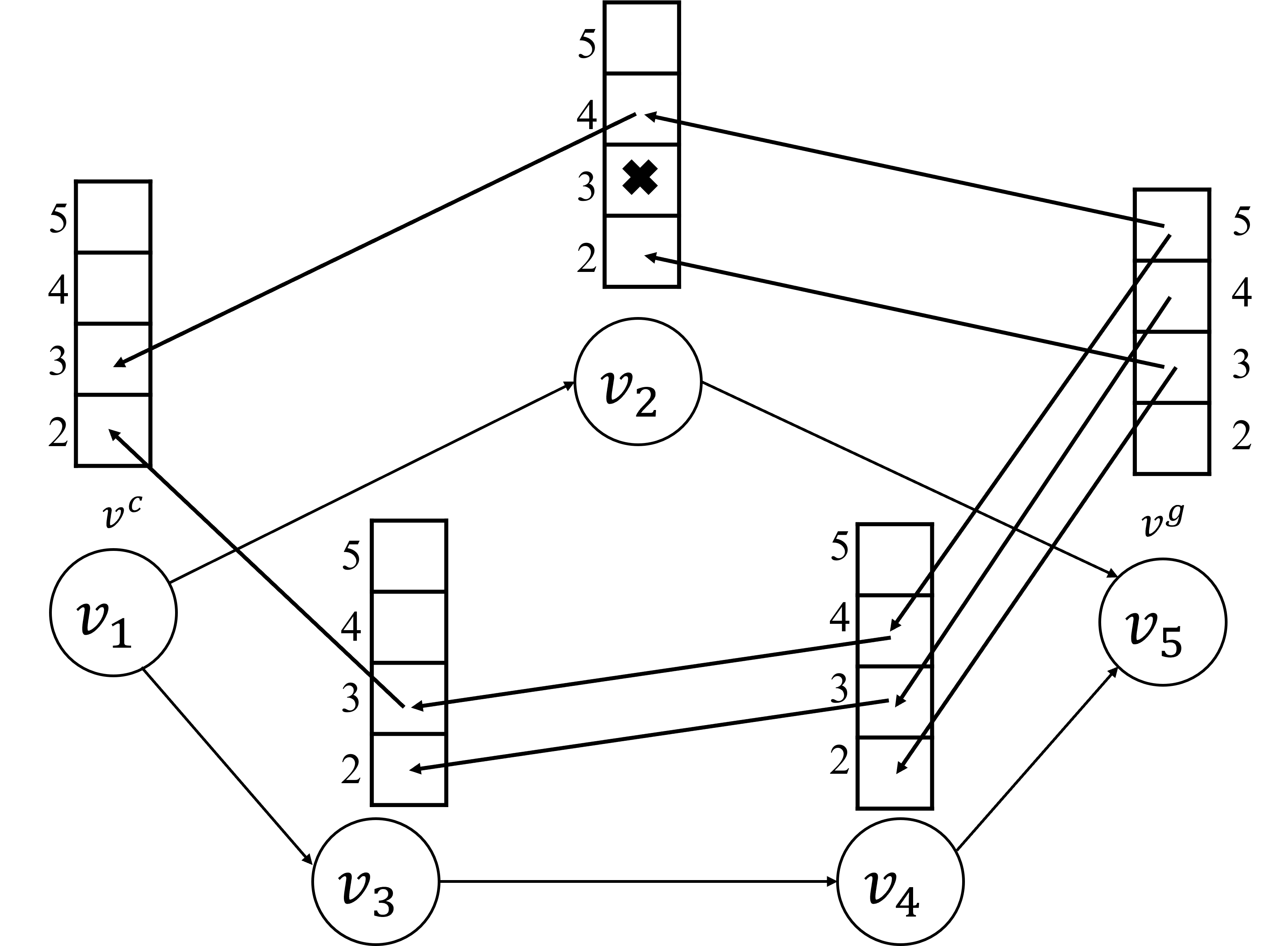}
    }
    \vfill
    \subfigure[Backward search on TIS states. The number in the block is the cost to $v^g$.]{
        \centering
        \label{fig:backward_search_tis}
        \includegraphics[width=0.33\textwidth]{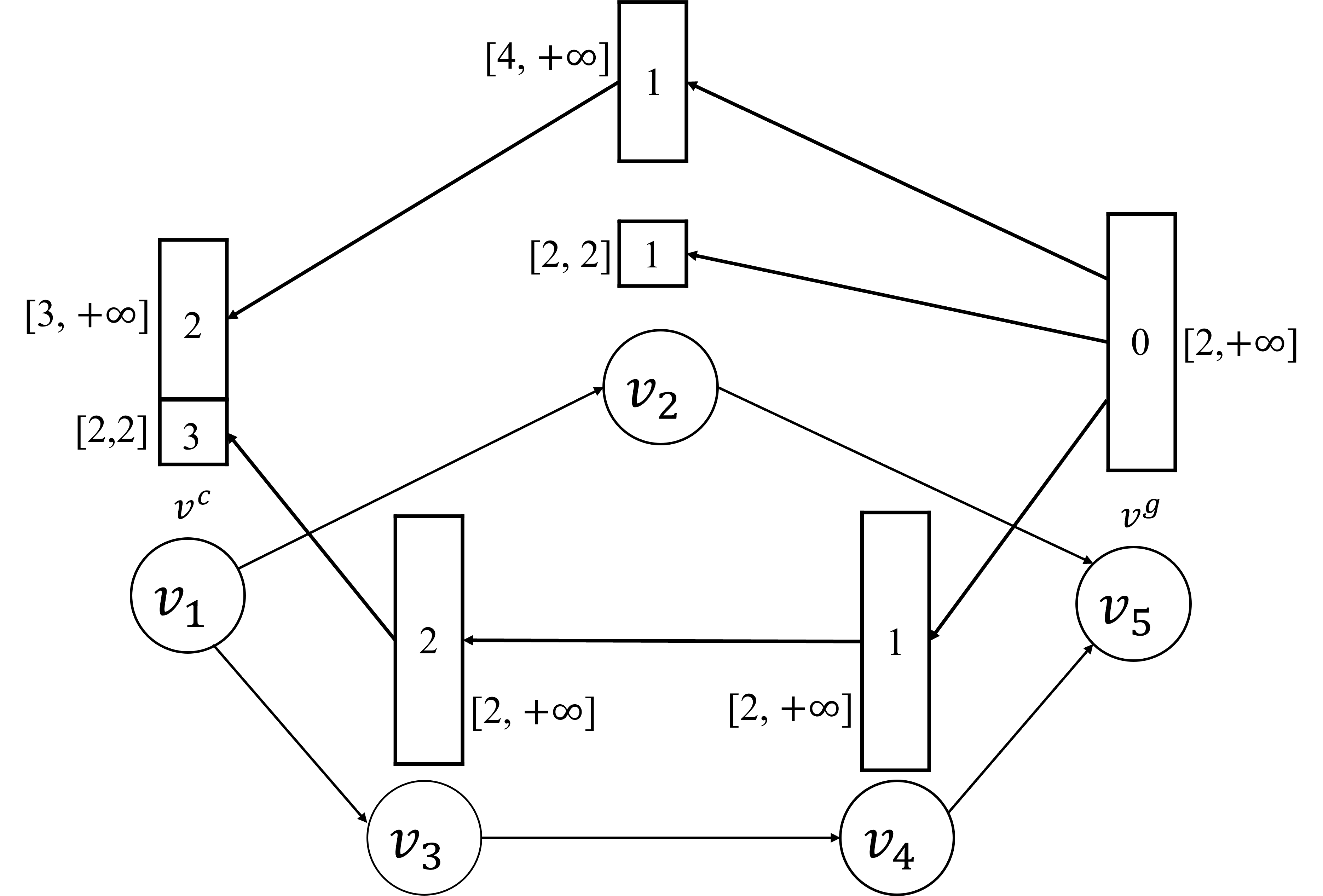}
    }
    \hspace{0.09\textwidth}
    \subfigure[Path building on TIS states.]{
        \centering
        \label{fig:build_path}
        \includegraphics[width=0.33\textwidth]{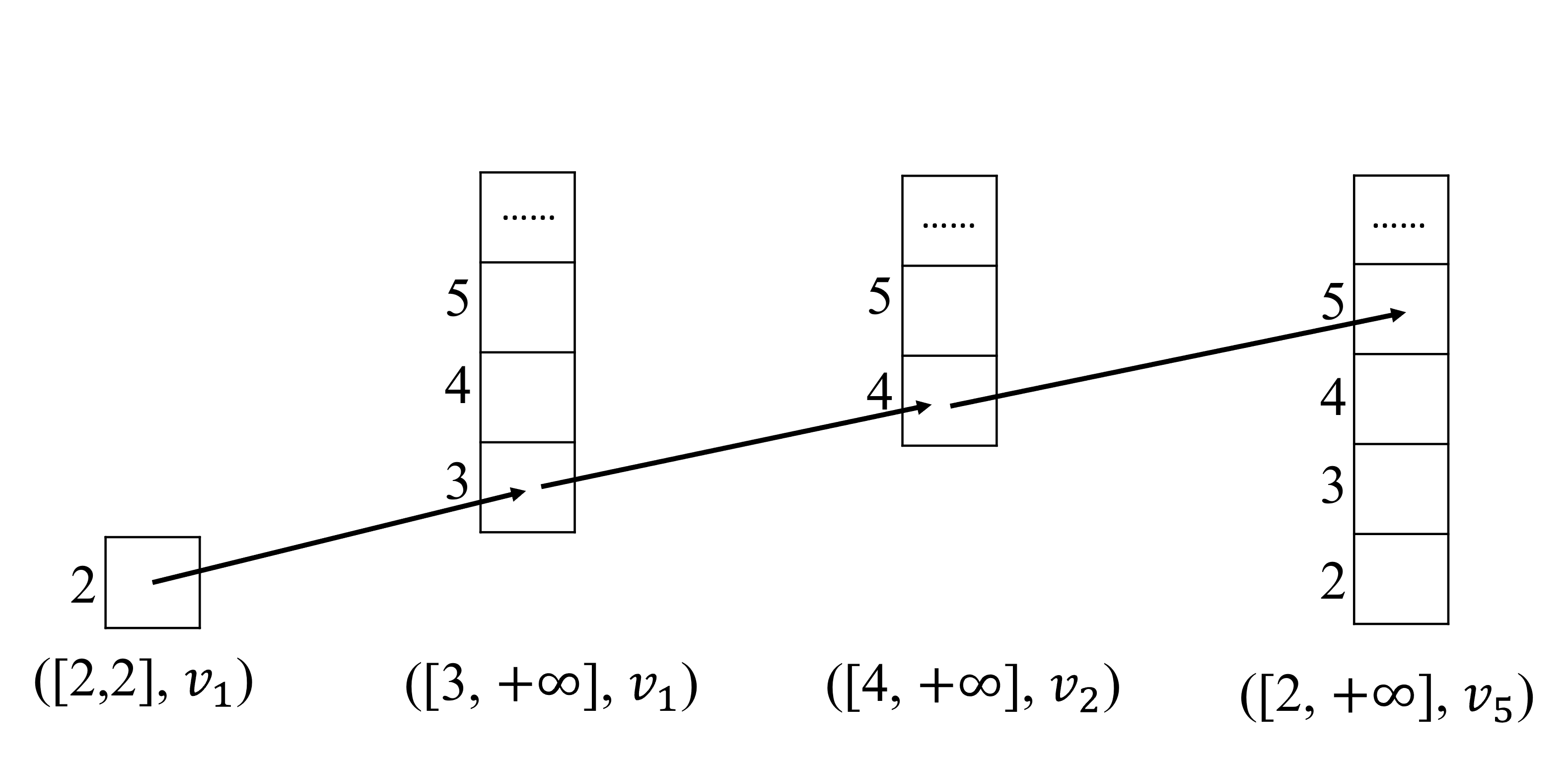}
    }\caption{(a)-(c) are three different search methods on the same graph and constraints. (d) is based on (c). The number near the block shows the time point or the time interval. }
\end{figure*}

\subsection{Sustainable Conflict-Based Search Algorithm}
The Sustainable Conflict-Based Search algorithm  (SCBS) calculates the multi-agent path plan and maintains the planning context sustainably. It is extended from the Conflict-Based Search algorithm (CBS)  \cite{sharon2015conflict}. The main difference between SCBS and CBS is the processing of the planning context. 

An example of using the SCBS algorithm is shown in Figure \ref{fig:system}.
Before each low-level search, the planning context directly related to the low-level search will be extracted according to the agent and its related constraints. We name this part of the planning context as individual planning context and use $ipc$ to represent it. During the low-level search, $ipc$ is modified to fit the new instance. After the low-level search, the new $ipc$ will be put back into $pc$ for usage in later iterations. 

Algorithm \ref{alg:scbs} shows the pseudo-code. In lines 5-7 and 25-27, $ipc$ is filtered from $pc$ by the function $GetIPC$. After using the low-level solver, the new $ipc$ is placed back to $pc$.

\begin{algorithm}[tb]
\caption{SCBS}
\label{alg:scbs}
\textbf{Input}: agents $A$, current time point $t^c$, planning context $pc$ 

\begin{algorithmic}[1] 
\STATE  $OPEN$ $\leftarrow$ $\emptyset$
\STATE  $R$ $\leftarrow$ new node
\STATE  $R.cons$ $\leftarrow$ $\emptyset$
\FOR{ each agent $a_i$ }
\STATE $ipc$ $\leftarrow$ $GetIPC$($pc$,$a_i$,$R.cons[a_i]$)
\STATE $R.path[a_i]$, $ipc$ $\leftarrow$ \\ \qquad $SRSIPP$($a_i$, $NULL$, $t^c$, $ipc$) // Algorithm 3
\STATE Update $pc$ by $ipc$.
\ENDFOR
\STATE $R.cost$ $\leftarrow$ calculate the SOC of $P.paths$
\STATE $OPEN$ $\leftarrow$ $OPEN$ $\cup$ \{$R$\}
\WHILE{$OPEN$ $\neq$ $\emptyset$}
\STATE $N$ $\leftarrow$ minimum cost node from $OPEN$.
\STATE $OPEN$ $\leftarrow$ $OPEN \backslash \{N\}$
\STATE $L$ $\leftarrow$ the earliest collision in $N$
\IF{$L$ is $None$}
\RETURN $N.paths$, $pc$
\ENDIF
\STATE $C$ $\leftarrow$ Get constraints from $L$
\FOR{constraint $c$ in $C$}
\STATE $P$ $\leftarrow$ new node
\STATE $P.cons$ $\leftarrow$ $N.cons$
\STATE $P.paths$ $\leftarrow$ $N.paths$
\STATE $a$ $\leftarrow$ $c.agent$
\STATE Insert $c$ in $P.cons[a]$.
\STATE $ipc$ $\leftarrow$ $GetIPC$($pc$, $a$, $P.cons[a]$)
\STATE $P.paths[a]$, $ipc$ $\leftarrow$ \\ \quad $SRSIPP$($a$, $P.cons[a]$, $t^c$, $ipc$)  // Algorithm 3
\STATE Update $pc$ by $ipc$.
\IF{$P.path[a]$ is not $NULL$}
\STATE $P.cost$ $\leftarrow$ calculate the SOC of $P.paths$
\STATE $OPEN$ $\leftarrow$ $OPEN$ $\cup$ \{$P$\}
\ENDIF
\ENDFOR
\ENDWHILE

\end{algorithmic}
\end{algorithm}

\subsection{Sustainable Reverse Safe Interval Path Planning Algorithm }
We now introduce the Sustainable Reverse Safe Interval Path Planning algorithm (SRSIPP). The SRSIPP is a single-agent solver based on A* \cite{hart1968formal} and SIPP \cite{phillips2011sipp}, designed for reusing the previous individual planning context to minimize the complexity of searching. We omit the agent index $i$ and the current number of the planning iteration $j$ in all symbols when discussing the SRSIPP algorithm, e.g. $t^s_{i,j}$, $t^g_{i,j}$, $v^s_i$, $v^g_i$ to $t^s$, $t^g$.
Let $v^c$ be the agent's current vertex, and $s^c = (t^c, v^c)$ be the agent's current state.

In the online MAPF, agents may replan while executing their plan. Although the starting vertex of each planning is different, the ending vertex is invariant. SRSIPP uses this property to achieve the target of reusing the previous planning context. For some single-agent solvers, the agent is planned from its current state to its goal through the edges. These search methods are called forward search. The planning can also search from the goal to its current state through the reverse edges. These search methods are called backward search. The SRSIPP is a backward search algorithm.



In the MAPF problem, most single-agent search methods are forward search on the time-space (TS) state. An example is shown in Figure \ref{fig:forward_search_ts}. However, since the entire search tree is rooted at the start vertex, which changes with each search, the forward search cannot reuse previous planning information. Observing that the goal vertex is invariant for the same agent, we can set the goal vertex as the root of the search tree to reuse this tree in the following search. However, we cannot predict the arrival time before the search starts. For the optimality of the algorithm, all states that reach the goal earlier must be fully searched before states that arrive later, resulting in a large amount of additional computation. Figure \ref{fig:backward_search_ts} shows an example.

To speed up the calculation, we propose to search on the time-interval-space (TIS) state. Figure \ref{fig:backward_search_tis} shows an example. Let $([t_l,t_r], v)$ be the TIS state where $[t_l, t_r]$ is a time interval and $v$ is the vertex, and $(t, v)$ be a TS state where $t$ is a time point. The TIS state $([t_l,t_r], v)$ contains a collection of TS states $\left \{ (t, v) | t\in [t_l, t_r] \right \}$.  Let $g_{TS}(s)$ and $g_{TIS}(s)$ be the cost from a TS and TIS state to $v^g$. All TS states in the same TIS state can use the same vertices sequence as the shortest path to the goal. Formally, we have
\begin{equation}
\label{equ:tists}
\forall t \in [t_l, t_r], g_{TIS}(([t_l,t_r], v)) = g_{TS}((t, v)).
\end{equation}
We refer to the function value of $g_{TS}(s)$ and $g_{TIS}(s)$ as the $g$ value of the TS state and the TIS state, respectively.
A TIS state is valid if and only if it does not cover any constrained TS states or cover TS states with a time point less than $t^s$. A maximum TIS state is a valid TIS state whose time interval is not a subset of other valid TIS states on the same vertex. Before searching a vertex, all the maximum TIS states in the vertex will be created. Their $g$ values are set to infinity, except that the $g$ values of TIS states on $v^g$ are set to $0$. 




The procedure of expanding a state for backward search is described as follows. The \textit{actions} of the search include moving to a neighbor vertex through reverse edges and staying in the same vertex. We use $([t_l, t_r], v)$ to represent the TIS state that needs to expand. We define a TIS state $s'$ as a dummy son of another TIS $s$ iff all TS states in $s'$ can take one action to one of TS states in $s$ in the forward direction. The cost of the dummy son will be one more than the original state, i.e., $g_{TIS}(s')=g_{TIS}(s)+1$. Let  $N(v) = \left \{ v' | ((v', v) \in E) \vee (v'=v) \right  \}$ be the neighborhood of $v$, and $ds(([t_l, t_r], v)) = \left \{([t_l', t_r'], v') | v' \in N(v) \right \}$ be the dummy son set of $([t_l, t_r], v)$. For the dummy son state $([t_l', t_r'], v') \in ds(([t_l, t_r], v))$.  We set 
\begin{equation}
\label{equ:tltr}
	\begin{split}
	t_l' &= \max(t^s, t_l - 1)  \\
	t_r' &=  t_r - 1
	\end{split}
\end{equation}

The dummy son is used to update all current TIS states on the same vertex. If the entire TIS state can be improved, the cost of the TIS state can be modified directly. Suppose only part of the TIS state can be improved due to the time interval coverage. In that case, the TIS state will be split into several TIS states according to the time interval coverage, and only the state completely covered by the dummy son's time interval will be updated. Figure \ref{fig:time_interval_updated} shows an example.

\begin{figure*}[t]
\centering
\includegraphics[width=0.75\textwidth]{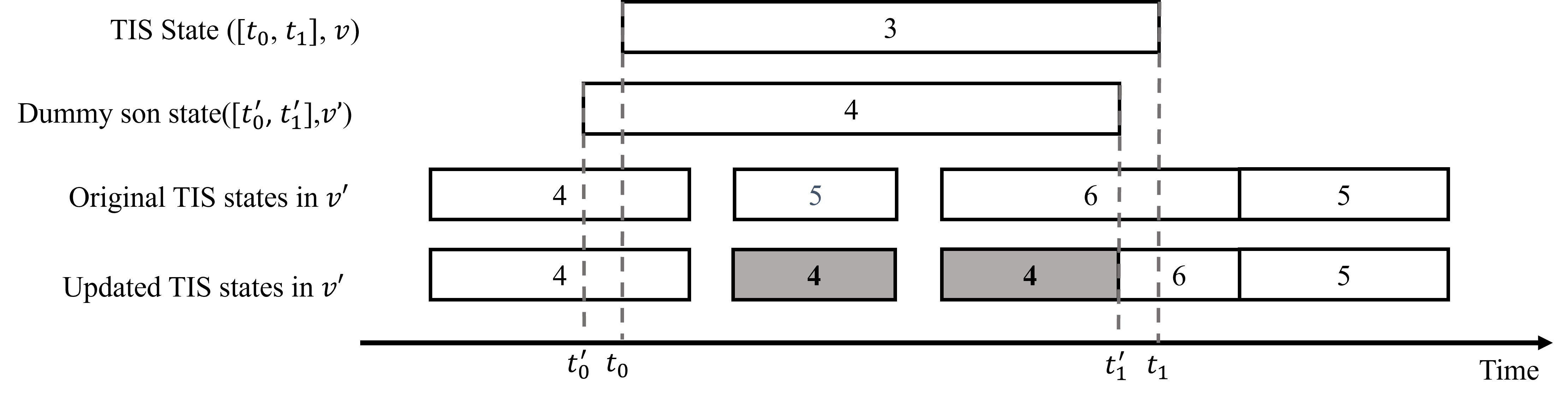} 
\caption{An example to show the process of updating the cost of TIS states. The blocks indicate the time interval of the state, and the number in the block shows the cost to $v^g$. The first row and the second row show the TIS state $([t_0, t_1],v)$ and its dummy son state $([t_0', t_1'],v)$. The third row indicates all valid TIS states in $v$ before being updated. The fourth line shows the updated TIS states. The first state cannot be improved, while for the second state, the whole state can be improved. For the third state, only part of the time interval is covered by $[t_0', t_1']$. The state is split into two states, and only the cost of the covered state is updated. The time interval of the fourth state is not covered, so it is not affected. The grey block shows the improved TIS states.}
\label{fig:time_interval_updated}
\end{figure*}


 
We use A* for the backward search. Let $h_{TIS}(s)$ be the heuristic function of the cost estimation from the TIS state $s$ to $s^c$, and $h_{v}(v)$ be the heuristic function of the path length estimation from vertex $v$ and $v^c$. We define $h_{TIS}(s)$ as
\begin{equation}
\label{equ:htis}
	h_{TIS}(([t_l, t_r], v)) = \max(\max(t_l-t^c,0), h_v(v))
\end{equation}
where $t_l \geq t^s$. The states, where $t_r < t^s$, will not be searched, and the heuristic function for these states is undefined.
In the 4-neighbor grid, $h_v(v)$ is usually defined by the Manhattan distance to the goal point, i.e.,
\begin{equation}
	h_{v}(v) = |v_x - v^g_x| + |v_y - v^g_y|
\end{equation}
where $(v_x, v_y)$ is the coordination of $v$ and $(v^g_x, v^g_y)$ is the coordination of $v^g$.
The evaluation function of a TIS state is defined as follows.
\begin{equation}
    f_{TIS}(s) = g_{TIS}(s) + h_{TIS}(s)
\label{equ:f}
\end{equation}
Let the function value of $h_{TIS}(s)$ and $f_{TIS}(s)$ be the $h$ value and the $f$ value of a TIS state $s$, respectively.

In the SRSIPP, the individual planning context includes the open list and the closed list. We use $OPEN$ and $CLOSED$ to represent them. At the beginning of the new search, we adjust the individual planning context to fit the new planning. Specifically, we recalculate the $h$ value and the $f$ value of all states in the $OPEN$, according to the new current state and Equations (\ref{equ:htis}, \ref{equ:f}). The $CLOSED$ can be used directly without modification. After the adjustment, the new search can reuse the 
$OPEN$ and the $CLOSED$ of the previous search. 


The pseudo-code is shown in Algorithm \ref{alg:srsipp}. In the code, $a.v^c$ and $a.v^g$ are the current vertex and the goal vertex of the agent $a$. In addition, $s.t_l$ and $s.t_r$ are the endpoints of the time interval in $s$. The $g$ value, $h$ value and $f$ value of state $s$ are saved in $s.g$, $s.h$ and $s.f$, respectively. A state is \textit{terminal state} if the state is the optimal final state for the search. We use $ts$ to save the terminal state and $cts$ to save all candidate terminal states. More specifically, $cts$ saves all closed states in the vertex $a.v^c$, which is reachable for the agent, i.e., $t^c \leq s.t_r$. In line 1, the individual planning context is extracted. If the current state is invalid, return directly (lines 2-4). In line 5, we update all states' $h$ value and $f$ value in $OPEN$. In line 6, if the TIS states on $a.v^g$ have not been created, create all maximum TIS states and put them into the $OPEN$. In lines 7-12, the function $StopCheck$ is used to check whether the search can stop. If yes, build the path by the function $BuildPath$ and return directly (We will discuss the detail of $StopCheck$ and $BuildPath$ later). Otherwise, the search starts. In each iteration, get the state with minimum $f$ value in the $OPEN$ (line 14). If the time interval cannot cover any time point after $t^c$, the state is useless for the current and later searches. We ignore it and go to the next iteration of the while loop (lines 15-17). In lines 18-27, dummy sons are generated to update the states. In line 29, we update the $OPEN$, the $CLOSED$, and the $cts$. Finally, we check whether the search can stop (lines 30-32). If no, go to the next iteration.
\begin{algorithm}[t]
\caption{SRSIPP}
\label{alg:srsipp}
\textbf{Input}:agent $a$, constraints $cons$, current time point $t^c$, individual planning context $ipc$

\begin{algorithmic}[1] 
\STATE $OPEN$, $CLOSED$ $\leftarrow$ $ipc$
\IF{$a$ is in the scene and ($t^c$, $a.v^c$) $\in$ $cons$}
\RETURN{$false, ipc$}
\ENDIF
\STATE Update $h$ value and $f$ value of states in $OPEN$.
\STATE Create maximum TIS states in $v'$ by $cons$ if the states are uncreated, and put them into $OPEN$.
\STATE $f_{min}$ $\leftarrow$ the smallest f value in $OPEN$
\STATE $cps \leftarrow \{([t_l, t_r], v)\in CLOSED| a.v^c=v  \wedge$ \\ $\qquad \qquad \qquad \qquad \qquad \qquad \qquad \quad a.t^s \leq t_r) \}$ 
\STATE $stop, ts \leftarrow StopCheck(cps, f_{min}, a)$ // Algorithm 4
\IF{$stop$} 
\RETURN{$BuildPath(ts, a), \{OPEN, CLOSED\}$}
\ENDIF
\WHILE{$OPEN$ is not empty}
\STATE $s$ $\leftarrow$ the TIS state with minimum $f$ value in $OPEN$
\IF{$s.t_r  < t^c$}
\STATE continue
\ENDIF

\FOR{ each vertex $v'$ in $N(s.v)$}
\STATE $ds$ $\leftarrow$ the dummy son of $s$ in $v'$
\STATE Create maximum TIS states in $v'$ by $cons$ if the states are uncreated, and put them into $OPEN$.
\FOR{ each TIS state $s'$ in $v'$ }
\IF{$s'$ can be improved by $ds$}
\STATE $S_{new}'$ $\leftarrow$ New states after improving $s'$
\STATE $OPEN \leftarrow OPEN \backslash \{s'\}  \cup S_{new}'$
\ENDIF
\ENDFOR
\ENDFOR
\STATE $stop, ts \leftarrow StopCheck(cps, s.f, a)$ // Algorithm 4
\STATE Update $OPEN$, $CLOSED$ and $cps$ by $s$.
\IF{$stop$}
\RETURN{$BuildPath(ts, a), \{OPEN, CLOSED\}$}
\ENDIF
\ENDWHILE

\end{algorithmic}
\end{algorithm}

The $StopCheck$ algorithm checks whether the search can stop and finds the best terminal state. 
When finding a better goal state out of $cps$ is impossible, we stop searching. There are two different scenarios for the stop. If the agent is in the scene, the search stops when a state in $cps$ covers $t^c$. Otherwise, the agent can choose a time point to enter the scene. Supposed a state $([t_l, t_r],v)$ is selected as the terminal state, the best enter time point is $\max(t_l, t^c)$, and the total cost from agents' current TS state $(t^c, v^c)$ to vertex $v^g$ is $\max(t_l - t_c, 0) + g_{TIS}(([t_l, t_r],v))$. If the minimum total cost of choosing a state in $cps$ is less than or equal to the $f$ value of the current expanded state, no better solution can be found, and the search can stop.

The pseudo-code of the $StopCheck$ algorithm is shown in Algorithm \ref{alg:ssc}. 
In lines 1-3, if there is no element in $cts$, the search can not stop. In lines 4-10, if the agent is in the scene, the search can stop only when a state in $cts$ covers $t^c$. In lines 12-18, if the agent is not in the scene, find the state with minimum total cost. If the cost is not higher than the minimum $f$ value of all nodes in the $OPEN$, the search can stop and return the best terminal state.

The $BuildPath$ function in Algorithm \ref{alg:srsipp} builds the final TS state path. After getting to the terminal state, we backtrack to get a TIS path. Based on it, we build the TS path as the solution. Specially, if the agent is not in the scene and the current time point is earlier than any time point in the terminal state's time interval, the agent will wait until it reaches the earliest time point of the target TIS state and then enter the scene. Figure \ref{fig:build_path} shows an example of the path building, selecting $([3, \infty], v_1)$ as the terminal state.

\section{Theoretical Analysis}
\begin{theorem}
\label{theo:theo1}
If $h_v(v)$ is admissible and satisfies the consistency assumption, when the first return value of the StopCheck function is true, it is impossible to have a better terminal state out of $cts$.
\begin{proof}
The proof is given in the appendix.
\end{proof}
\end{theorem}
\begin{theorem}
If $h_v(v)$ is admissible and satisfies the consistency assumption, the SRISPP algorithm is complete and optimal.
\label{theo:theo5}
\end{theorem}
\begin{proof}
If $h_v(v)$ is admissible and satisfies the consistency assumption, then $h_{TIS}(s)$ is also admissible and satisfies the consistency assumption. It is proved in the appendix.

If it is the first planning for the configuration of $a$ and $cons$, the $OPEN$ and $CLOSED$ are empty initially. It is a standard A*, and the search's completeness and optimality are satisfied.

If the $OPEN$ and $CLOSED$ are not empty at the beginning of the search, the state in the $CLOSED$  will not be reopened, and the $g$ value of the state is already the smallest cost to the goal vertex. For the state in the $OPEN$, we update their $h$ value and $f$ value to fit the new situation. The scenario can be seen starting from a snapshot in the standard A*, except that some states are closed in advance. However, the early closed nodes will not affect the completeness and optimality of the algorithm because all their unclosed neighbors are in the $OPEN$.

Combining theorem \ref{theo:theo1} and the above discussions, the theorem is proved.

\end{proof}

\begin{corollary}
If $h_v(v)$ is admissible and satisfies the consistency assumption, SCBS is complete and optimal.
\end{corollary}
\begin{corollary}
If $h_v(v)$ is admissible and satisfies the consistency assumption, SR is complete and snapshot optimal.
\end{corollary}
These two corollaries can be proved according to the property of CBS, and RA algorithm  \cite{sharon2015conflict,vsvancara2019online}. Because Theorem \ref{theo:theo5} is proved and the operations on planning context in SCBS and SR will not affect completeness and optimality, the corollaries are proved.

\begin{algorithm}[tb]
\caption{StopCheck}
\label{alg:ssc}
\textbf{Input}: candidate terminal state set $cts$, the minimum estimated function value in the open list $f_{min}$, the agent $a$ 

\begin{algorithmic}[1] 
\IF{$cts$ is empty}
\RETURN{$false, NULL$}
\ENDIF

\IF{$a$ is in the scene}
\STATE{s $\leftarrow$ the state which covers $t^c$} 
\IF{$s$ is $NULL$}
\RETURN{$false, NULL$}
\ELSE
\RETURN{$true, s$}
\ENDIF
\ELSE
\STATE{$s_{min} \leftarrow argmin_{s \in cts}(\max(s.t_l - t^c, 0) + s.g$)}
\STATE{$c_{min} \leftarrow min_{s \in cts}(\max(s.t_l - t^c, 0) + s.g$)}
\IF{$c_{min} \leq f_{min}$}
\RETURN{$true, s_{min}$}
\ELSE
\RETURN{$false, NULL$}
\ENDIF 
\ENDIF

\end{algorithmic}
\end{algorithm}

\begin{figure}[t]
\centering
    \subfigure[Small grid maps.]{
       \centering
       \label{fig:small_grids}
        \includegraphics[width=0.20\textwidth]{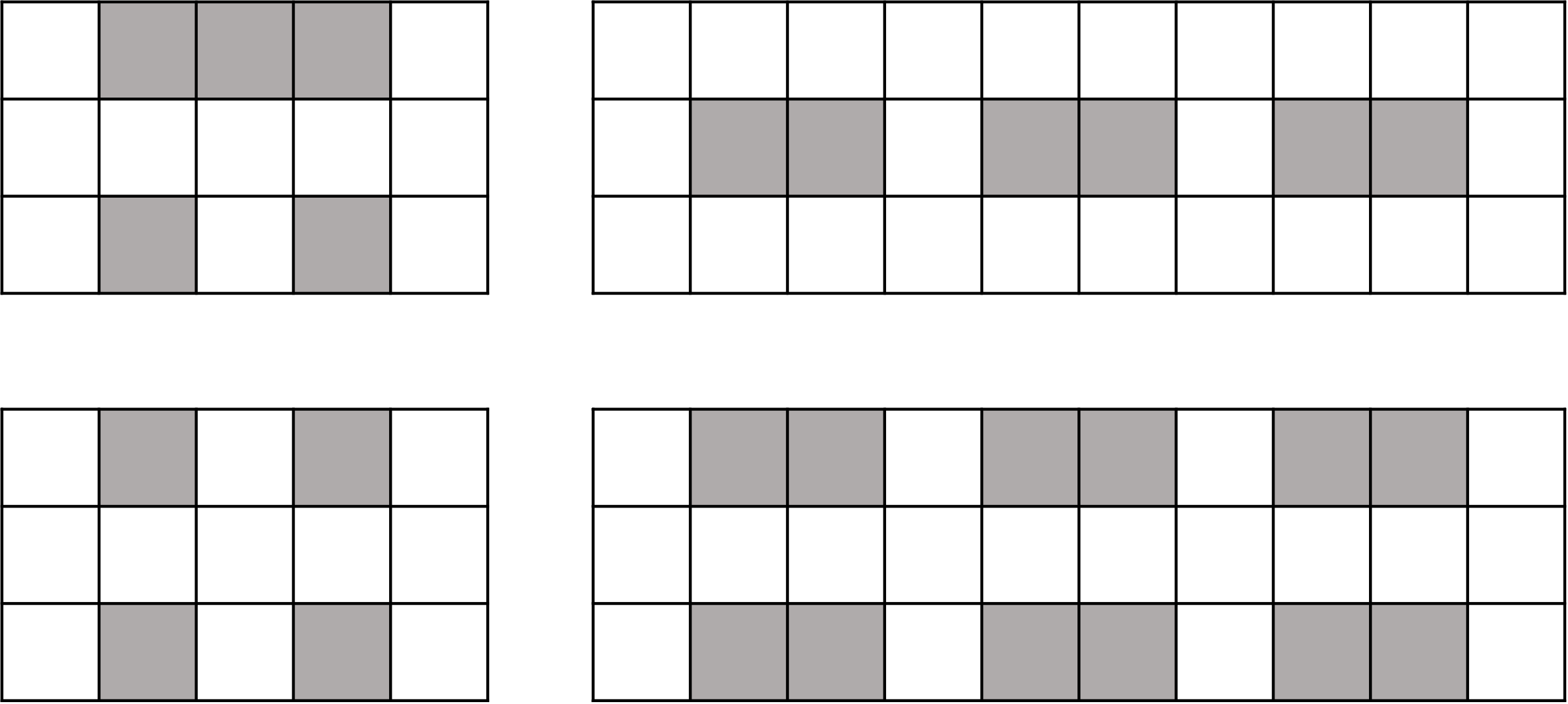}
    }\subfigure[Large grid map.]{
        \centering
        \label{fig:large_grids}
        \includegraphics[width=0.20\textwidth]{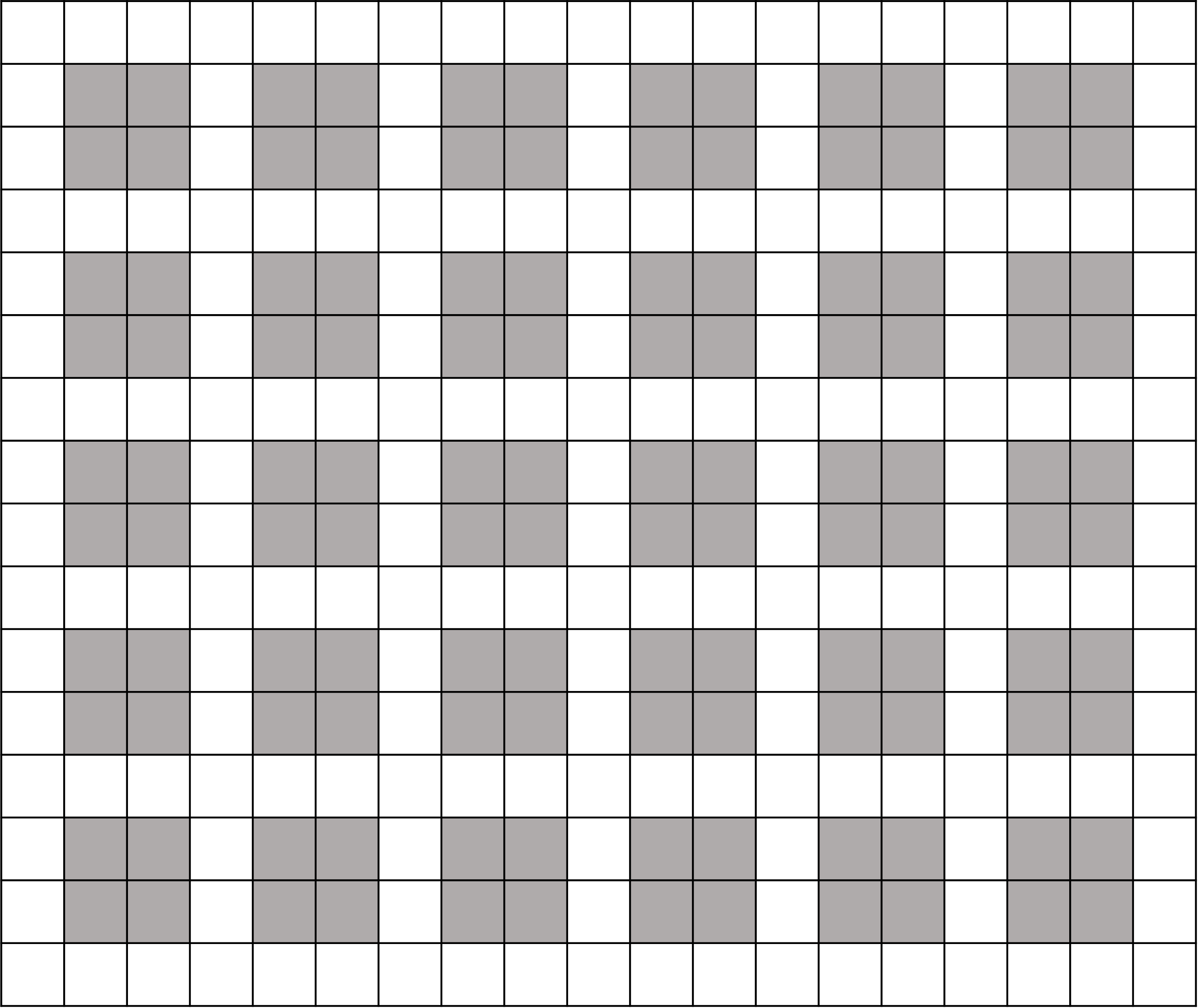}
    }\caption{Grids for experiments from \cite{vsvancara2019online}.}
\end{figure}

\section{Experiment}
The purpose of our experiments is to evaluate the computational efficiency of the proposed approach. 
Two 4-neighbor grid map datasets are used with settings similar to \cite{vsvancara2019online}, which is a small grid map dataset and a large grid map dataset, respectively. We perform online MAPF in these grid maps. Specifically, we move each of the agents from one cell to another. During this period, there will be new agents starting at any time.


We use success rate and average running time as the metrics. The running time limit is 30 seconds. If the algorithm run exceeds the time limit, the experimental instance is unsuccessful and uses the time limit as the running time. We make statistics based on the number of agents. For each number of agents, it has 100 instances in both datasets.

The experiments assume no agents are in the scene at the beginning. On the one hand, the offline parts of algorithms are the same. If some instances fail in the offline part, the online algorithm is not executed. These test cases are useless for comparison. All randomly generated instances are always solvable if there are no agents in the scene in the beginning. On the other hand, all compared methods use identical offline MAPF solvers. Ignoring it does not affect the comparison. 

We run the algorithms on an AMD R7-5800X CPU with 4.40 GHz and 16GB RAM. Four baselines are selected for comparison:
\begin{itemize}
\item \textbf{RA+CBS+A*(A1)}: This algorithm uses the Replan All algorithm to solve the online MAPF problem. When new agents start, it uses the CBS algorithm to calculate the multi-agent path plan, whose low-level solver is A*.

\item \textbf{RA+CBS+RSIPP(A2)}: This algorithm removes all sustainable operations from our proposed approach, i.e., all solvers do not maintain or use the planning context.

\item \textbf{SR+SCBS+RSIPP(A3)}: For this algorithm, the low-level solver can not reuse the planning context. In the middle-level solver, if the agent strictly follows the previous path in the node of the conflict tree, the low-level solver won't be called, and a suffix from the results of the previous planning is used as the results of this planning.

\item \textbf{SR+SCBS+SRSIPP(A4)}: The proposed method.
\end{itemize}

\subsection{Small Grid Map}
In the first dataset, 4 small grid maps are used, as shown in Figure \ref{fig:small_grids}. The start and goal points are randomly sampled in two cells of the opposing sides of the grids, and the start time point is uniformly sampled from $[1, 30]$. Let $k$ be the number of agents, which is in the range $\{10, 12, 15, 17, 20, 22, 25\}$. In each grid map, we generate 25 experimental instances for each 
setting of the agent number. For each number of agents, it has $25*4=100$ instances.

Table \ref{table:small_grid_succ} shows the result of success rate. Except when the agent number is 17, the success rate of A4 is larger than or equal to A1. Table \ref{table:small_grid} shows the running time of all baselines and the speedup ratio relative to A1 of other baselines. In the result, A2 did not obtain significant improvement, while A3 can speed up in most cases. The best performance comes from A4. It shows an improvement in average runtime relative to A1 under all agent number settings and achieves the maximum speedup ratio in all baselines.

\begin{table}[t]
\centering
\small
\begin{tabular}{c|cccc}
$k$ &A1 & A2 &A3 & A4\\\hline
10&\textbf{96\%} &\textbf{96\%} &\textbf{96\% }&\textbf{96\%}  \\ 
12&\textbf{87\%} &84\% &86\% &\textbf{87\%}  \\ 
15&78\% &74\% &79\% &\textbf{80\%}  \\ 
17&\textbf{65\%} &63\% &63\% &64\%  \\ 
20&\textbf{49\%} &47\% &47\% &\textbf{49\%}  \\ 
22&36\% &35\% &37\% &\textbf{37\%}  \\ 
25&27\% &26\% &28\% &\textbf{29\%}  \\ 
\end{tabular}
\caption{Table of the success rate in small grids. $k$ is the agent number.}
\label{table:small_grid_succ}
\end{table}

\begin{table}[t]
\centering
\small
\begin{tabular}{c|cccc}
$k$ &A1 & A2 &A3 & A4\\\hline
10&1.68(-)&2.2(0.77)&1.72(0.98)&\textbf{1.58(1.06)} \\ 
12&5.65(-)&5.96(0.95)&5.35(1.06)&\textbf{4.77(1.18)} \\ 
15&8.05(-)&8.67(0.93)&7.96(1.01)&\textbf{7.27(1.11)} \\ 
17&12.26(-)&12.61(0.97)&11.96(1.03)&\textbf{11.54(1.06)} \\ 
20&17.01(-)&17.18(0.99)&16.6(1.02)&\textbf{16.23(1.05)} \\ 
22&21.08(-)&22.02(0.96)&20.83(1.01)&\textbf{20.15(1.05)} \\ 
25&22.53(-)&22.39(1.01)&22.15(1.02)&\textbf{21.9(1.03)} \\ 
\end{tabular}
\caption{Table of the running time and the speedup ratio in small grids. The first number in the cell shows the running time(sec), and the number in parentheses indicates the speedup relative to A1. $k$ is the agent number.}
\label{table:small_grid}
\end{table}

\begin{table}[t]
\centering
\small
\begin{tabular}{c|cccc}
$k$ &A1 & A2 &A3 & A4\\\hline
90&84\% &85\% &88\% &\textbf{91\%}  \\ 
92&82\% &84\% &86\% &\textbf{89\%}  \\ 
94&90\% &89\% &91\% &\textbf{92\%}  \\ 
96&75\% &77\% &81\% &\textbf{84\%}  \\ 
98&72\% &76\% &82\% &\textbf{83\%}  \\ 
100&54\% &69\% &72\% &\textbf{77\%}  \\ 
\end{tabular}
\caption{Table of the success rate in large grids. $k$ is the agent number.}
\label{table:large_grid_succ}
\end{table}

\begin{table}[ht]
\centering
\small
\begin{tabular}{c|cccc}
$k$ &A1 & A2 &A3 & A4\\\hline
90&10.09(-)&8.98(1.12)&7.48(1.35)&\textbf{6.06(1.67)} \\ 
92&9.69(-)&9.61(1.01)&8.13(1.19)&\textbf{6.8(1.43)} \\ 
94&8.04(-)&8.27(0.97)&6.59(1.22)&\textbf{5.52(1.46)} \\ 
96&13.05(-)&13.35(0.98)&11.25(1.16)&\textbf{9.84(1.33)} \\ 
98&13.37(-)&12.4(1.08)&10.5(1.27)&\textbf{9.21(1.45)} \\ 
100&18.06(-)&14.78(1.22)&13.02(1.39)&\textbf{11.58(1.56)} \\ 
\end{tabular}
\caption{Table of the running time and the speedup ratio in large grids. The first number in the cell shows the running time(sec), and the number in parentheses indicates the speedup relative to A1. $k$ is the agent number.}
\label{table:large_grid}
\end{table}

\subsection{Large Grid Map}
In the second dataset, we use the large grid map shown in Figure \ref{fig:large_grids}. We generate $100$ experimental instances for each 
setting of the agent number $k$ in the range of $[90, 100]$. The start time points are randomly sampled from $[1, 100]$. An agent's start point and goal point are sampled from two different sides.

Table \ref{table:large_grid_succ} and Table \ref{table:large_grid} show the result of the success rate, the average running time, and the speedup ratio relative to A1 of different baselines. A2 and A3 perform better than A1 in most cases, while A4 runs the fastest among all methods. In all instances, the average speedup ratio of A4 relative to A1 reaches $1.48$. 

The running time of $92$ and $94$ agents are shorter than $90$ agents. It is because when the number of agents is large, the running time is not significantly influenced by the small increment of the agent number but is dominated by some hard-to-solve cases.

\subsection{Discussion}

The improvement in the large grids is more obvious than in the small grids. We believe it is because of the path length. In the small grid maps, the path of agents is short. Sustainable information can only be used a few times. In the large map, the path is much longer. Sustainable information is used more frequently, making A4 get a higher speedup ratio. 

\section{Conclusion}
We proposed a three-level algorithm to solve the online MAPF problem. Three levels are responsible for simulating the multi-agent online environment, solving the multi-agent path planning, and using the historical planning information to assist in solving the single-agent path planning. We proved the completeness and the snapshot optimality of our approach. The experiment shows that our proposed method runs faster than the SOTA algorithm. During the experiment, we also found that the performance in large grids is much better than in small grids. This is because the agent has a longer path in a larger grid so that the planning context can be reused more times. It shows that the longer the path, the better the acceleration effect of our method.

In the future, more aspects of using sustainable information, such as building the conflict tree, will be explored to improve the efficiency of the algorithm further.

\bibliography{aaai23}

\end{document}